\documentclass[12pt]{article}

\usepackage{amsfonts}
\usepackage{amssymb}
\usepackage{amsmath}

\newtheorem{theorem}{Theorem}[section]
\newtheorem{proposition}[theorem]{Proposition}

\newtheorem{lemma}[theorem]{Lemma}

\newtheorem{defin}[theorem]{Definition}
\newenvironment{definition}{\begin{defin}\em}{\end{defin}}
\newtheorem{defins}[theorem]{Definitions}

\newenvironment{proof}{\noindent \textbf{Proof: }}{\hfill
$\Box$  \vspace{1ex}}
\newtheorem{ex}[theorem]{Examples}

\newtheorem{exs}[theorem]{Example}

\newtheorem{rem}[theorem]{Remark}

\newtheorem{rems}[theorem]{Remarks}

\newtheorem{corollary}[theorem]{Corollary}

\def\cocoa{{\hbox{\rm C\kern-.13em o\kern-.07em C\kern-.13em o\kern-.15em A}}}

\def\F{\mathbb{F}}

\def\N_0{\mathbb{N}_0}

\def\lm{\textrm{lm}}

\def\ba{{\boldsymbol{a}}}

\def\ll{\prec}

\def\a{\alpha}

\def\fq{\mathbb{F}_q}

\def\fqx{\mathbb{F}_q[\mathbf{X}]}

\begin{document}

\begin{center}
{\Large\textbf{On the second Hamming weight of some Reed-Muller type codes}}
\end{center}
\vspace{3ex}

\noindent
\textsc{C\'{\i}cero  Carvalho}\footnote{The author is partially supported  by  CNPq grants 302280/2011-1 and   470416/2011-4, and by FAPEMIG proc.\  PPM-00127-12\ \ \  email: cicero@ufu.br \\
To appear in Finite Fields and Their Applications.} 
\\
{\em\small Faculdade de Matem\'atica, Universidade Federal de Uberl\^andia,  Av.\ J.\ N.\ \'Avila 2121, 38.408-902 - Uberl\^andia - MG, Brazil.}
\vspace{1ex}

\noindent
{\footnotesize \textbf{Abstract.} 
We study affine cartesian codes, which are a Reed-Muller type of evaluation codes, where polynomials are evaluated at the  cartesian product of $n$ subsets of a finite field $\mathbb{F}_q$. These codes appeared recently in a work by  H.\ L\'opez, C.\ Renter\'ia-Marquez and R.\ Villareal (see \cite{lrv}) and, independently, in a generalized form, in a work by O.\ Geil and C.\ Thomsen (see  \cite{gt}). Using a proof technique developed by O. Geil (see \cite{geil}) we  determine the second Hamming weight (also called next-to-minimal weight) for  particular cases of affine cartesian codes and also some higher Hamming weights of this type of code.
 }
\vspace{2ex}

\vspace{2ex}

\noindent
\small{\textbf{AMS Classification:} 11T71, 13P25, 94B60, 

%

\section{Introduction}

Affine variety codes are evaluation codes which were introduced by J. Fitzgerald and R. F. Lax in \cite{f-lax} and their construction is as follows. Let $I \subset \fq[X_1, \ldots, X_n] =: \fqx$ be an ideal and set $I_q := I + (X_1^q - X_1, \ldots, X_n^q - X_n)$. Then the affine variety $V_{\fq}(I)$ defined by $I$  in $\fq^n$ coincides with the affine variety $V_{\overline{\fq}}(I_q)$ defined by $I_q$ in ${\overline{\fq}}^n$ (where $\overline{\fq}$ denotes an algebraic closure of $\fq$). 
Let $V_{\fq}(I) = \{P_1, \ldots, P_m\}$ and denote by $\varphi : \fqx/I_q \rightarrow \fq^m$  the evaluation morphism $\varphi(f + I_q) = (f(P_1), \ldots, f(P_m))$. 

\begin{definition}
Let $L$ be an $\F_q$-vector subspace of $\fqx/I_q$. The {\em affine variety code} $C(L)$ is the image $\varphi(L)$.
\end{definition}

In \cite{lrv}  L\'opez, Renter\'ia-Marquez and Villareal defined affine cartesian codes, a special type of affine variety codes,  in the following way. Let $A_1, \dots, A_n$ be non-empty subsets of $\F_q$ and let $X := A_1 \times \cdots \times A_n \subset \mathbb{F}_q^n$.
Let $f_i = \prod_{c \in A_i} (X_i - c)$ for $i = 1, \ldots, n$ and let $I = (f_1, \ldots, f_n)$,  clearly the set of zeroes of $I$ is $X$. Furthermore, 
$f_i$ is a factor of $X_i^q - X_i = \prod_{c \in \mathbb{F}_q} (X_i - c)$ for all $i = 1, \ldots, n$ so $I = I_q$. From $X_i^q \equiv X_i \pmod{I}$ for all $i = 1, \ldots, n$ we get $f^q \equiv f \pmod{I}$ for any $f \in \mathbb{F}_q[X_1, \ldots, X_n]$ hence $I$ is radical: in fact, if $f^r \in I$ then $f^s \in I$, where $s \in \{0, \ldots, q - 1\}$ is such that $s \equiv r \pmod{q}$, so that $f ^q \in I$ and a fortiori $f \in I$.  A similar reasoning shows that the ideal generated by $I$ in $\overline{\mathbb{F}_q}[X_1, \ldots, X_n]$ is radical so from Nullstellensatz $I$ is the ideal of the set $X$ (this was proved in a different way in \cite[Lemma 2.3]{lrv}).

\begin{definition}
Let $d$ be  a positive integer, the {\em affine cartesian code} $C(d)$ is the image, by $\varphi$, of the classes in 
$\fqx/I$
of the zero polynomial and of polynomials having degree up to $d$. 
\end{definition} 

A very important particular case of such codes is of course when $A_i = \fq$ for all $i = 1, \ldots, n$, for then we have the so-called generalized Reed-Muller codes. 

Let $d_i := \#(A_i)$ for $i = 1, \ldots, n$, 
in their study of affine cartesian codes L\'opez et al.\  proved that we may assume $2 \leq d_1 \leq \cdots \leq d_n$ and that the dimension of  $C(d)$ is equal to 
\begin{equation*}
\begin{split}
  & \binom{n + d }{d}  -   \sum_{i = 1}^n \binom{n + d - d_i}{d - d_i} + \cdots +   \\ & (-1)^j \sum_{1 \leq i_1 < \cdots < i_j \leq n} \binom{n + d - d_{i_1} - \cdots -   d_{i_j}}{d - d_{i_1} - \cdots -   d_{i_j}} + \cdots +   (-1)^n \binom{n + d - d_{1} - \cdots -   d_{n}}{d - d_{1} - \cdots -   d_{n}} 
\end{split}
\end{equation*}
where we set $\binom{a}{b} = 0$ if $b < 0$ (see \cite[Thm. 3.1]{lrv}).
The length of $C(d)$ is clearly $d_1 . \cdots . d_n$.  In \cite{lrv} it is also proved that the minimum distance $d_{\min}(C_d)$ of $C(d)$ is equal to  1 if $d \geq \sum_{i =1}^n (d_i -1)$, and there is  a formula for $d_{\min}(C_d)$ when $1 \leq d < \sum_{i =1}^n (d_i -1)$. 

In the next section we will determine the exact value of 
the second Hamming weight, also called next-to-minimal weight, for some particular cases of $C(d)$  as well as some   higher Hamming weights of these codes (see Theorems \ref{thm2} to \ref{thm3} and Corollary \ref{cor-last}). 
In the case of generalized Reed-Muller codes, the study of the values for the second Hamming weight was started by  J.-P. Cherdieu and
R. Rolland (see \cite{ch-rol}), and the complete determination of these values has been recently done by A.\ A.\  Bruen (see \cite{brue}). Bruen  discovered that these weights had already been determined in the Ph.D.\ thesis of D.\ Erickson (see \cite{eric}) for many values of $d$, and showed how the  remaining cases can be obtained from earlier works by him. The values of the weights for these remaining cases also follow from results in \cite{geil} or in \cite{rr}. The characterization of the second weight codewords of generalized Reed-Muller codes has just  been  completed by E.\ Leducq (see \cite{elodie} and the references therein for earlier results on this subject).

\section{Main  results}

Given an ideal $J \subset  \fqx$ and a monomial order $\prec$ in the set of monomials of $\fqx$ we denote by $\Delta(J)$ the {\em footprint} of $I$ with respect to $\prec$, i.e. $\Delta(J)$ is the set of monomials  in $\fqx$ which are not leading monomials of polynomials in $J$. From the definition of Gr\"obner basis (with respect to $\prec$) we get that a monomial is in $\Delta(J)$ if and only if it is not a multiple of any of the leading monomials of the polynomials in a Gr\"obner basis for $J$.  If $J = (g_1,\ldots, g_r)$ and we denote by $\Delta( \lm(g_1), \ldots, \lm(g_r) )$ the set of monomials of $\fqx$ which are not a multiple of the leading monomial of $g_i$ for all $i \in \{1, \ldots, r\}$ then $\Delta(J) \subset \Delta( \lm(g_1), \ldots, \lm(g_r) )$ (we will use this fact in the proofs of Proposition \ref{thm1} and Theorem \ref{thm2}). A well-known property of the footprint is that the classes of the elements of $\Delta(J)$ are a basis for $\fqx/J$ as an $\fq$-vector space (see e.g.\ \cite[Prop.\ 6.52]{becker}).  Also, when $\Delta(J)$ is a finite set we get  $\#(V_{\overline{\fq}}(J)) \leq \#(\Delta(J))$, and equality holds when $J$ is a radical ideal (see \cite[Thm. 8.32]{becker}).

In what follows we will use the graded-lexicographic order $\ll$ which is defined on the monomials of $\fqx$ by setting  $X_1^{m_1} \cdots X_n^{m_n} \ll X_1^{t_1} \cdots X_n^{t_n}$ if and only if $\sum_{i = 1}^n m_i < \sum_{i = 1}^n t_i$ or, if $\sum_{i = 1}^n m_i = \sum_{i = 1}^n t_i$, then the  leftmost nonzero entry in $(t_1 - m_1, \ldots, t_n - m_n)$ is positive. Observe that using this order we get $\lm(f_i) = X_i^{d_i}$ for all $i = 1, \ldots, n$ and since $X_i^{d_i}$ and $X_j^{d_j}$ are relatively prime for all distinct $i, j \in \{1. \ldots, n\}$ we have that $\{f_1, \ldots, f_n\}$ is a Gr\"obner basis for $I = I_q$ (see \cite[Prop.\ 4, p.\ 104]{iva}), thus
\[
\Delta(I) = \{ X_1^{ a_1} . \cdots . X_n^{a_n} \; | \; 0 \leq a_i < d_i \; \forall\; i = 1, \ldots, n\}.
\]
Given $F \in \fqx$ let $R \in \fqx$ be its remainder in the division by $\{f_1, \ldots, f_n\}$, then $\varphi(F + I) = \varphi(R + I)$ and from the division algorithm we get that $\deg R \leq \deg F$. This shows that  
\[
C(d) = \varphi(\langle \Delta(I)_{\leq d} \rangle)
\]
where $\langle \Delta(I)_{\leq d} \rangle$ is the $\mathbb{F}_q$-vector space generated by $\Delta(I)_{\leq d} = \{ M \in \Delta(I) \; | \; \deg(M) \leq d\}$. We note that this gives a proof that $d_{\min}(C(d)) = 1$ if $d \geq \sum_{i = 1}^n (d_i - 1)$: in fact one may show that there are polynomials $F_1, \ldots, F_m$ (where $m := d_1. \cdots . d_n$) such that $F_i(P_j) = \delta_{i j}$ for all $i, j \in \{1, \ldots, m\}$ (see e.g. \cite[p. 406]{dgm}) so that $\varphi$ is surjective and a fortiori an isomorphism because $\dim_{\mathbb{F}_q} (\fqx/I) = \#(\Delta(I)) = m$, so from $\Delta(I) = \Delta(I)_{\leq d} $ for all $d \geq \sum_{i = 1}^n (d_i - 1)$ we get $C(d) =\mathbb{F}_q^m$ for all $d \geq \sum_{i = 1}^n (d_i - 1)$.

We will need the following two lemmas in the proof of the main results.

\begin{lemma} \label{lema1} Let $ 0 < d_1 \leq \cdots \leq d_n$ and $0 \leq s \leq \sum_{i = 1}^n (d_i - 1)$ be integers. Let $m(a_1, \ldots , a_n) = \prod_{i = 1}^n (d_i - a_i)$, where $0 \leq a_i < d_i$ is  an integer for all $i = 1,\ldots, n$. Then
\[ 
\min \{ m(a_1, \ldots, a_n) \, | \, a_1 + \cdots + a_n \leq s \} =  (d_{k + 1} - \ell) \prod_{i = k + 2}^n d_i
\]
where $k$ and $\ell$ are uniquely defined by $s = \sum_{i =1}^k (d_i -1) + \ell $,  with $0 \leq \ell < d_{k + 1} - 1$ (if $s < d_1 - 1$ then take $k = 0$ and $\ell = s$, if $k + 1 = n$ then we understand that $\prod_{i = k + 2}^n d_i = 1$). 
\end{lemma}
\begin{proof}
We start by observing that the minimum must be attained when $\sum_{i = 1}^n a_i = s$. Thus, let  $\ba = (a_1, \ldots, a_n)$, with $\sum_{i = 1}^n a_i = s$ be such that $a_{i_1} < d_{i_1} - 1$ for some $i_1 \in\{1, \ldots, n\}$. If there exists $i_2 \in \{1, \ldots , n\}$ such that $i_1 < i_2$, $a_{i_2} > 0$  and $a_{i_1} + a_{i_2} \leq d_{i_1} - 1$, then denoting by $\ba'$  the $n$-tuple obtained from $\ba$ by replacing $a_{i_1}$ by $a_{i_1} + a_{i_2}$ and $a_{i_2}$ by 0, we get that 
\[
m(\ba) - m(\ba') = (a_{i_1} a_{i_2} + (d_{i_2} - d_{i_1})a_{i_2})  \prod_{\stackrel{i = 1}{i\neq i_1, \, i_2}}^n (d_i - a_i) \geq 0
\] 
so that $m(\ba') \leq m(\ba)$, and note that $m(\ba') < m(\ba)$ if $a_{i_1} > 0$. If there exists $i_2 \in \{1, \ldots , n\}$ such that $i_1 < i_2$, $a_{i_2} > 0$  and $a_{i_1} + a_{i_2} > d_{i_1} - 1$, then denoting by $\ba''$  the $n$-tuple obtained from $\ba$ by replacing $a_{i_1}$ by $d_{i_1} - 1$ and $a_{i_2}$ by $a_{i_2} - (d_{i_1} - a_{i_1} -  1)$ we get that 
\[
m(\ba) - m(\ba'') = (d_{i_1} - a_{i_1} - 1)(d_{i_2} - a_{i_2} - 1) \prod_{\stackrel{i = 1}{i\neq i_1, \, i_2}}^n (d_i - a_i) \geq 0
\]
so that $m(\ba'') \leq m(\ba)$. This proves that  $m$ attains its minimum at $\ba = (a_1, \ldots, a_n)$  where $a_i = d_i - 1$ for $i \in \{1, \ldots, k\}$, $a_{k + 1} = \ell$ and $a_j = 0$ for $j > k + 1$.
\end{proof}

\begin{lemma} \label{lema2}
Let $\,2 \leq s \leq  d_1 \leq \cdots \leq d_n$ be integers, with $n \geq 2$. Let $q(a_1, \ldots , a_n) = \prod_{i = 1}^n (d_i - a_i)$ where $0 \leq a_i < s$ is  an integer for all $i = 1,\ldots, n$. Then
\[
\min \{ q(a_1, \ldots, a_n) \, | \, a_1 + \cdots + a_n \leq s \} =  (d_1 - (s -1)) (d_2 -1) \prod_{i = 3}^n d_i.
\]
\end{lemma}
\begin{proof} As in the previous Lemma we observe that the minimum must be attained when $\sum_{i = 1}^n a_i = s$. Thus, let  $\ba = (a_1, \ldots, a_n)$, with $\sum_{i = 1}^n a_i = s$ and assume that $a_1 < s - 1$. If there exists $i_2 \in \{2, \ldots, n\}$ such that $a_{i_2} > 0$ and $a_1 + a_{i_2} \leq s- 1$ then denoting by $\ba'$  the $n$-tuple obtained from $\ba$ by replacing $a_1$ by $a_1 + a_{i_2}$ and $a_{i_2}$ by 0, we get that 
\[
m(\ba) - m(\ba') = (a_1 a_{i_2} + (d_{i_2} - d_1)a_{i_2}) \prod_{\stackrel{i = 2}{i\neq  i_2}}^n (d_i - a_i) \geq 0
\]
so   $m(\ba) \geq m(\ba')$ and  $m(\ba) > m(\ba')$ if $a_1 \neq 0$. If there exists $i_2 \in \{2, \ldots, n\}$ such that $a_1 + a_{i_2} > s - 1$ then we must have $a_1 > 0$ and $a_{i_2} = s - a_1$,  denoting by $\ba''$  the $n$-tuple obtained from $\ba$ by replacing $a_1$ by $s - 1$ and $a_{i_2}$ by 1 we get 
\[
m(\ba) - m(\ba'') =  (d_{i_2} - d_1 + a_1 - 1)(s - a_1 - 1) \prod_{\stackrel{i = 2}{i\neq  i_2}}^n (d_i - a_i) \geq 0.
\]
This shows that if $q$ attains its minimum at $\ba = (a_1, \ldots, a_n)$ then we may assume that $a_1 = s - 1$ and now it is easy to check that we can also assume $a_2 = 1$.
 \end{proof}

As mentioned in the introduction, in \cite{lrv} the authors find a formula for the minimum distance of affine cartesian codes is determined (see \cite[Thm. 3.8]{lrv}).  The determination of this formula occupies most of the paper, the result being preceded   by several technical  lemmas. In following we present a simple proof of this result which we will use in the main results. We also note that in \cite[Prop.\ 5]{gt} there is  a formula for the minimum distance of certain codes which may be seen as a generalization of affine cartesian codes.

\begin{proposition}\label{thm1} The minimum distance of the affine cartesian
code $C(d)$ defined over $X = A_1 \times \cdots \times A_n$, with $d_i := \#(A_i)$ for all $i = 1,\ldots, n$, is $d_{\min} = (d_{k + 1} - \ell)  \prod_{i = k + 2}^n d_i$, where $d < \sum_{i = 1}^n (d_i - 1)$ and $k$ and $\ell$ are uniquely defined by $d = \sum_{i = 1}^k (d_i -1) + \ell$, with $0 \leq \ell < d_{k + 1} - 1$ (if $d < d_1 - 1$ then take $k = 0$ and $\ell = s$, if $k + 1 = n$ we understand that $\prod_{i = k + 2}^n d_i = 1$).
\end{proposition}
\begin{proof}Let $F \in \fq[X_1,\ldots, X_n]$ be  a polynomial 
which is a sum of monomials in $\Delta(I)_{\leq d}$ 	
and let $J_F := (F, f_1, \ldots, f_n)$,  the weight of the codeword $\varphi(F + I)$ then satisfies $w(\varphi(F + I)) = \prod_{i = 1}^n d_i - \#(V_{\mathbb{F}_q} (J_F))$. Since $\#(V_{\mathbb{F}_q} (J_F)) \leq \#(\Delta(J_F)) \leq \#( \Delta(\lm(F), X_1^{d_1}, \ldots, X_n^{d_n}))$ we get that $\prod_{i = 1}^n d_i - \#( \Delta(\lm(F), X_1^{d_1}, \ldots, X_n^{d_n}))$ is a lower bound for $w(\varphi(F + I))$. 
Let $\lm(F) = X_1^{a_1} . \cdots . X_n^{a_n}$, 
from  $\#(\Delta(\lm(F), X_1^{d_1},\ldots, X_n^{d_n}) = \prod_{i = 1}^n d_i -  \prod_{i = 1}^n (d_i - a_i)$ we get $w(\varphi(F + I)) \geq \prod_{i = 1}^n (d_i - a_i)$.  Letting $(a_1,\ldots, a_n)$ run over all $n$-tuples such that $\sum_{i = 1}^n a_i \leq d$ we get	 from Lemma \ref{lema1} that $(d_{k + 1} - \ell)  \prod_{i = k + 1}^n d_i$ is a lower bound for the minimum distance of $C(d)$. To see that this lower bound is attained we write $A_i = \{ \a_{i 1}, \ldots, \a_{i d_i} \}$ for all $i = 1, \ldots, n$ and let $G(X_1, \ldots, X_n) = \prod_{i = 1}^n \prod_{j = 1}^{a_i} (X_i - \a_{i j})$, then $G(X_1, \ldots, X_n)$ is a polynomial with leading monomial equal to $X_1^{a_1}.\cdots . X_n^{a_n}$ which has $\prod_{i = 1}^n d_i -  \prod_{i = 1}^n (d_i - a_i)$ zeroes, all in $X$.
\end{proof}

We will now determine the second Hamming weight of codes $C(d)$ for several particular cases of this code. We start with the case where all the sets in the cartesian product have the same cardinality $a$ and $2 \leq d < a$ (hence $a \geq 3$). The proof of the following theorem is an enhancement of the proofs of \cite[Prop.\ 2 and Thm.\ 3]{geil}.

\begin{theorem}\label{thm2} Let $A_i \subset \mathbb{F}_q$ such that $\#(A_i) = a  \geq 3$ for all $i = 1, \ldots, n$, with $n \geq 2$ and let $2 \leq d < a$. The second Hamming weight of $C(d)$ is $(a - (d - 1))(a -1) a^{n - 2}$. 
\end{theorem}
\begin{proof} We write $A_i = \{ \a_{i 1}, \ldots, \a_{i a} \}$ for all $i = 1, \ldots, n$, and let $1 \leq t < a$.  Let $F \in \mathbb{F}_q[X_1, \ldots, X_n]$ 
be a polynomial of degree $t$ 
and let $J_F = (F, f_1, \ldots, f_n)$. 
As in the proof of Proposition \ref{thm1} we have that $w(\varphi(F + I) ) = \prod_{i = 1}^n d_i  - \#(V_{\mathbb{F}_q}(J_F) )$. Let $M := X_1^{a_1} . \cdots . X_n^{a_n}$ be the leading monomial of $F$ (so that $\sum_{i = 1}^n a_i = t$ because we are using the graded-lexicographic order). We deal first with the case where $t \geq 2$.\\
a) Assume that $a_i < t$ for all $i = 1, \ldots, n$. From 
$$\#(V_{\mathbb{F}_q}(J_F) ) \leq \#(\Delta(J_F) ) \leq \#( \Delta(M, X_1^{d_1}, \ldots, X_n^{d_n}) ) = \prod_{i = 1}^n d_i - \prod_{i = 1}^n (d_i - a_i) $$ 
and Lemma \ref{lema2} we get $w(\varphi(F + I)) \geq (d_1 - (t -1)) (d_2 -1) \prod_{i = 3}^n d_i$.  This bound is effectively attained, for example, when $F = \left( \prod_{i = 1}^{t - 1} (X_1 - \a_{1 i}) \right)  (X_2 - \a_{2 1})$. \\ 
b)  Assume now that $a_j = t$ for some $j \in \{ 1, \ldots, n\}$. If  $\{F, f_1, \ldots, f_n\}$ is a Gr\"obner basis for 
$J_F$ then $\#( \Delta(J_F) ) = t a^{n - 1}$ and  $w(\varphi(F + I) ) = a^n - t a^{n - 1} = (a - t)a^{n - 1}$;   
from  Proposition \ref{thm1} we get that this is the minimum distance of $C(t)$. If  $\{F, f_1, \ldots, f_n\}$ is not a Gr\"obner basis for  $J_F$ then the $S$-polynomial $S(F, f_j) = X_j^{a - t} F - f_j$ must have a nonzero remainder $R$ in the division by $\{F, f_1,\ldots, f_n\}$ (otherwise $\{F, f_1, \ldots, f_n\}$ would be a Gr\"obner basis because any other pair of distinct polynomials $\{g_1, g_2 \}$ in $\{F, f_1, \ldots, f_n\}$ has leading monomials which are relatively prime - see \cite[pags.\ 103 and 104]{iva}). Let $L := X_1^{b_1} . \cdots . X_n^{b_n}$ be the leading monomial of $R$, from the division algorithm we get $b_j < t$, $b_i < a$  for all $i \in \{1, \ldots, n\}$, $i \neq j$ and $\sum_{i = 1}^n b_i \leq \deg(S(F,f_j)) \leq a$.  Thus $J_F = (F, f_1, \ldots, f_n) = (R, F, f_1, \ldots, f_n) $ so that 
\[ 
\#(\Delta(J_F)) \leq  \#(\Delta(L, X_j^t, X_1^a, \ldots, X_n^a) ) = t a^{n - 1} - (t - b_j) \prod_{i = 1, i \neq j}^n(a - b_i) 
\]
Now we apply Lemma \ref{lema1} with $d_1 = t$, $d_i = a$ for $i = 2,\ldots, n$ and $s = a$, and 
writing $a = (t - 1) + (a - (t - 1))$ we get that an upper bound for 
the number of zeroes of $F$ in $X$ is $t a^{n - 1} - (t - 1) a^{n - 2}$  so the minimum distance of $\varphi(F + I)$ is lower bounded by $a^n - t a^{n - 1} + (t - 1) a^{n - 2} = (a - 1)(a - t + 1) a^{n - 2}$. This proves that for $2 \leq t < a$  the possible values for $w(F + I)$, where $F$ is a polynomial of degree $t$ are in the set $\{(a - t)a^{n - 1} \} \cup \{ w \in \mathbb{N} \; | \, w \geq (a - 1)(a - t + 1) a^{n - 2} \}$ where $(a - t)a^{n - 1}$ and $(a - 1)(a - t + 1) a^{n - 2}$ are realized as weights.

In the case where $t = 1$ we have  $M = X_j$ for some $j \in \{1, \ldots, n\}$ so that 
$\#(\Delta(M, X_1^a, \ldots, X_n^a) = a^n - (a - 1)a^{n -1}$, thus 
$w(F + I) \geq (a - 1)a^{n -1}$.

Now we put the above results together to calculate the second smallest weight of $C(d)$, where $2 \leq d < a$, and find that it is equal to $(a - 1)(a - d + 1) a^{n - 2}$. This is because $(a - 1)(a - d + 1) a^{n - 2} < (a - 1)(a - t + 1) a^{n - 2}$ and $(a - 1)(a - d + 1) a^{n - 2} < (a - t) a^{n - 1}$ for all 
$1 \leq t < d$, and of course $(a - d)a^{n - 1} <(a - 1)(a - d + 1) a^{n - 2}$.
\end{proof}

Setting  $a = q$ in the above theorem we get the values for the second Hamming weight of the generalized Reed-Muller codes when $2 \leq d < q$ (cf.\ \cite{geil}).

In the next theorem we treat the case where we have the cartesian product of  two subsets of $\fq$ with distinct cardinalities.

\begin{theorem}\label{tm2.5} Let $A_1, A_2 \subset \fq$ be such that  $3 \leq \#(A_1) =: d_1 < d_2 := \#(A_2)$ and let $2 \leq d < d_1$. The second Hamming weight of $C(d)$ is $(d_1 - d + 1))(d_2 -1)$. 
\end{theorem}
\begin{proof}
We follow the same procedure of the above proof, and although the beginning is similar the development is a bit more elaborate. 
We write $A_i = \{ \a_{i 1}, \ldots, \a_{i d_i} \}$ for $i = 1, 2$, and let $1 \leq t < d_1$.  Let $F \in \mathbb{F}_q[X_1, X_2]$ 
be a polynomial of degree $t$ 
and let $J_F = (F, f_1, f_2)$. Then  
$w(\varphi(F + I) ) \geq d_1 d_2  - \#(\Delta(J_F) )$. Let $M := X_1^{a_1} .  X_2^{a_2}$ be the leading monomial of $F$ (hence $a_1 + a_2 = t$). We deal first with the case where $t \geq 2$.\\
a) Assume that $a_i < t$ for $i = 1, 2$. From $\#(\Delta(J_F) ) \leq \#( \Delta(M, X_1^{d_1}, X_2^{d_2}) ) = d_1 d_2 - \prod_{i = 1}^2 (d_i - a_i) $ and Lemma \ref{lema2} we get $w(\varphi(F + I)) \geq (d_1 - (t -1)) (d_2 -1)$.  This bound is effectively attained, for example, when $F = \left( \prod_{i = 1}^{t - 1} (X_1 - \a_{1 i}) \right)  (X_2 - \a_{2 1})$. \\ 
b) Assume now that $a_j = t$ for $j = 1$ or $j = 2$. If  $\{F, f_1,  f_2\}$ is a Gr\"obner basis for 
$J_F$ then $\#( \Delta(J_F) ) = t d_2$, if $a_1 = t$ or $\#( \Delta(J_F) ) = t d_1$, if $a_2 = t$ so that  $w(\varphi(F + I)) \geq d_1 d_2 - t d_2$ if $a_1 = t$ or $w(\varphi(F + I)) \geq d_1 d_2 - t d_1$ if $a_2 = t$. According to Proposition \ref{thm1} $(d_1 - t) d_2$  is the minimum distance of $C(t)$, and it is easy to check that $(d_2 - t) d_1$ is also realized as the weight of a codeword. We assume now that $\{F, f_1,  f_2\}$ is not a Gr\"obner basis for 
$J_F$, and we treat separatedly the cases where $M = X_1^t$ and $M = X_2^t$. 

When $M = X_1^t$ we must have that the $S$-polynomial $S(F, X_1) = X_1^{d_1 - t}F - X_1^{d_1}$ has a nonzero remainder in the division by $\{F, X_1^{d_1}, X_2^{d_2} \}$ (because $X_1^t$ and $X_2^{d_2}$ are relatively prime), so let $L := X_1^{b_1} X_2^{b_2}$ be the leading monomial of this remainder. From the division algorithm we get $b_1 < t$, $b_2 < d_2$ and $b_1 + b_1 \leq d_1$.  
We have $\#( \Delta(J_F) )  \leq \#( \Delta(L, M, X_1^{d_1}, X_2^{d_2}) ) = t d_2 - (t - b_1)(d_2 - b_2)$ so $w(\varphi(F + I)) \geq d_1 d_2 - t d_2 + (t - b_1)(d_2 - b_2)$. We now use Lemma \ref{lema1} to find the minimum of $(t - b_1)(d_2 - b_2)$, observing the restrictions on $b_1 $ and $b_2$, and get  
$w(\varphi(F + I)) \geq d_1 d_2 - t d_2 + d_2 - d_1 + t - 1 = (d_2 - 1)(d_1 - t + 1)$.

When $M = X_2^t$ we have that the $S$-polynomial $S(F, X_2) = X_2^{d_2 - t}F - X_2^{d_2}$ has a nonzero remainder in the division by $\{F, X_1^{d_1}, X_2^{d_2} \}$ and again we denote by $L = X_1^{b_1} X_2^{b_2}$ the leading monomial of this remainder. From the division algorithm we get  $b_1 < d_1$, $b_2 < t$ and $b_1 + b_2 \leq d_2$, but from $b_1 < d_1$ and $b_2 < t$ we also get $b_1 + b_2 \leq d_1 + t - 2$, thus $b_1 + b_2 \leq  r := \min \{d_2, d_1 + t - 2 \}$. As before we note that $\#( \Delta(J_F) )  \leq \#( \Delta(L, M, X_1^{d_1}, X_2^{d_2}) ) = t d_1 - (d_1 - b_1)(t - b_2)$ so that 
$w(\varphi(F + I)) \geq d_1 d_2 - t d_1 + (d_1 - b_1)(t - b_2)$. Now we want to apply Lemma \ref{lema1} to find the minimum of $(t - b_2)(d_1 - b_1)$, observing the restrictions on $b_1$ and $b_2$. If $r = d_1 + t - 2$ then from $d_1 + t - 2 = (t - 1) + (d_1 - 1)$ we get that the minimum is 1, hence  
$w(\varphi(F + I)) \geq d_1 (d_2 - t) + 1$. If $r = d_2$ then $d_2 \leq d_1 + t - 2$ so $d_2 - t + 1 \leq d_1 - 1$, thus from $d_2 = (t - 1) + d_2 - t + 1$ and Lemma \ref{lema1} we get that the minimum is  $d_1 - d_2 + t - 1$, which implies that  $w(\varphi(F + I)) \geq (d_1 - 1)(d_2 - t + 1)$.

This completes the analysis of the case where $t\geq 2$. In the case where $t = 1$ we have that  either $w(\varphi(F + I)) \geq (d_1 - 1) d_2$ or $w(\varphi(F + I)) \geq d_1 (d_2 - 1)$.  

From what is done so far we get that if $2 \leq t < d_1$ then $ w(\varphi(F + I)) \in \{(d_1 - t) d_2\} \cup \{ v \in \mathbb{N} \, | \, v \geq (d_2 - 1)(d_1 - t + 1) \}$ because $ (d_2 - 1)(d_1 - t + 1) - d_1(d_2 - t)  = -(t - 1)(d_2 - d_1 - 1) \leq 0$ and $ (d_2 - 1)(d_1 - t + 1) -  (d_1 - 1)(d_2 - t + 1) = -(t - 2)(d_2 - d_1) \leq 0$. 

Thus considering the weights $w(\varphi(F + I))$ for all polynomials $F$  of degree less of equal than $d$ (where $2 \leq d < d_1$) we get that the second smallest weight is $(d_2 - 1)(d_1 - d + 1)$, this is because   $(d_2 - 1)(d_1 - d + 1) < (d_2 - 1)(d_1 - t + 1)$ and  
$(d_2 - 1)(d_1 - d + 1) < (d_1 - t) d_2$ whenever $1 \leq t < d$, and $(d_1 - d) d_2 < (d_2 - 1)(d_1 - d + 1)$.  
\end{proof}

The following result deals with higher Hamming weights of the code $C(d)$.
The proof is an enhancement of the proof of \cite[Thm.\ 4]{geil}.

\begin{theorem}\label{thm3}
Let $2 \leq d_1 \leq \cdots \leq d_n$ be integers, with $n \geq 2$, and let $d$ be an integer such that $\sum_{i = 1}^{n - 1} (d_i - 1) \leq d < \sum_{i = 1}^n (d_i - 1)$. Write $d = \sum_{i = 1}^{n - 1} (d_i - 1) + \ell$, with $0 \leq \ell < d_{n} - 1$. Then for $t \in \{1, \ldots, \ell + 1\}$ the $t$-th weight of $C(d)$ is $d_n - \ell + (t - 1)$.
\end{theorem}

\begin{proof} For $t \in \{1, \ldots, \ell + 1\}$ we have $C(d - (t - 1)) \subset C(d)$ so from Proposition \ref{thm1} we get that in $C(d)$ there are words of weight $d_n - \ell, d_n - \ell + 1, \ldots, d_n$, being $d_n - \ell$ the minimum distance of $C(d)$. This proves the theorem.
\end{proof}

We now put the last three results together to determine the second Hamming weight of $C(d)$, for all $d \geq 2$, in the case where we have the cartesian product of two sets containing at least three elements each.

\begin{corollary}\label{cor-last}
Let $A_1, A_2 \subset \fq$ be such that  $3 \leq \#(A_1) =: d_1 \leq d_2 := \#(A_2)$ and let $2 \leq d$. Then  second Hamming weight of $C(d)$ is equal to: \\
i) $(d_1 - d + 1)(d_2 - 1)$ if $2 \leq d < d_1$; \\
ii) $d_1 + d_2 - d$ if $d_1 \leq d \leq d_1 + d_2 - 2$; \\
iii) 2 if $d_1 + d_2 - 2 < d$.  
\end{corollary}
\begin{proof}
Item (i) is a direct consequence of Theorems \ref{thm2} and \ref{tm2.5}. Item (ii) is a consequence of the above theorem, because writing $d = (d_1 - 1) + \ell$ we get that the second weight is $d_2 - \ell + 1 = d_1 + d_2 - d$. Item (iii) comes from the fact that  $C(d) =\mathbb{F}_q^m$ whenever $d \geq d_1 + d_2 - 2$ as observed just before Lemma \ref{lema1} (this is also proved in \cite{lrv}).
\end{proof}

\end{document}